\documentclass[a4paper,11pt]{article}
\usepackage{fullpage}

\pdfoutput=1
\usepackage[colorlinks=true,urlcolor=Blue,citecolor=Green,linkcolor=BrickRed]{hyperref}
\usepackage[usenames,dvipsnames]{xcolor}
\usepackage[T1]{fontenc}
\usepackage[utf8]{inputenc}
\usepackage{amssymb,amsmath,amsfonts,amsthm}
\usepackage{algorithm}
\usepackage{algorithmicx}
\usepackage[noend]{algpseudocode}
\usepackage[noadjust]{cite}
\usepackage[OT4]{fontenc}
\usepackage{amsthm}
\usepackage{amssymb}
\usepackage{amsmath}
\usepackage{amsfonts}
\usepackage{todonotes}
\usepackage{rotating}
\usepackage{xspace}
\usepackage{verbatim}
\usepackage{mathrsfs}
\usepackage{framed}
\usepackage{thm-restate}
\usepackage[capitalise]{cleveref}
\usepackage{colortbl}
\usepackage{array,multirow}
\usepackage{authblk}

\def\dd{\mathinner{.\,.}}
\newcommand{\Oh}{\mathcal{O}}

\DeclareMathOperator{\poly}{poly}

\newcommand{\set}[1]{\left\lbrace #1 \right\rbrace}

\let\Cref\cref
 
\newtheorem{theorem}{Theorem}[section]
\newtheorem{@theorem}{Theorem}[section]

\newtheorem{lemma}{Lemma}[section]

\crefname{conjecture}{Conjecture}{Conjectures}
\crefname{lemma}{Lemma}{Lemmas}
\crefname{problem}{Problem}{Problems}
\crefname{remark}{Remark}{Remarks}
\crefname{definition}{Definition}{Definitions}
\crefname{observation}{Observation}{Observations}
\crefname{@theorem}{Theorem}{Theorems}
\crefname{fact}{Fact}{Facts}
\crefname{claim}{Claim}{Claims} 

\title{Better Indexing for Rectangular Pattern Matching\footnote{Partially supported by the Polish National Science Centre grant number 2023/51/B/ST6/01505.}}

\author[1]{Paweł Gawrychowski}
\author[1]{Adam Górkiewicz}
\affil[1]{Institute of Computer Science, University of Wrocław, Poland}

\date{}

\begin{document}

\maketitle

\begin{abstract}
We revisit the complexity of building, given a two-dimensional string of size $n$, an indexing structure that allows
locating all $k$ occurrences of a two-dimensional pattern of size $m$. While a structure of size $\Oh(n)$ with query time $\Oh(m+k)$
is known for this problem under the additional assumption that the pattern is a square [Giancarlo, SICOMP 1995], a popular
belief was that for rectangular patterns one cannot achieve such (or even similar) bounds, due to a lower bound for a certain natural
class of approaches [Giancarlo, WADS 1993]. We show that, in fact, it is possible to construct
a very simple structure of size $\Oh(n\log n)$ that supports such queries for any rectangular pattern in $\Oh(m+k\log^{\varepsilon}n)$ time,
for any $\varepsilon>0$. Further, our structure can be constructed in $\tilde{\Oh}(n)$ time.
\end{abstract}

\section{Introduction}
In the area of algorithms on strings, two basic algorithmic questions are pattern matching and string indexing. In the former, we aim
to locate an occurrence of the pattern in the text, while in the latter the goal is to preprocess the text for multiple such queries.
The complexity of both problems is well understood, at least for the case of regular strings and exact occurrences. In particular,
a text of length $n$ can be indexed in $\mathcal{O}(n)$ space, so that all $k$ occurrences of a pattern of length $m$ can be retrieved
in $\Oh(m+k)$ time: this is a textbook application of suffix trees, already explained in the original article by Weiner~\cite{DBLP:conf/focs/Weiner73}.

In this paper, we consider two-dimensional strings, which are simply (two-dimensional) arrays of characters. Such a generalisation
is naturally motivated by possible applications in image processing, and the complexity of pattern matching for two-dimensional
strings has been already considered in the 70s~\cite{DBLP:journals/siamcomp/Baker78a,Bird1977}, and further investigated
in the early 90s~\cite{DBLP:journals/siamcomp/AmirBF94,DBLP:journals/siamcomp/GalilP96}. A common assumption in all of these
papers is that both the text and the pattern are square arrays, but this assumption is just to avoid clutter and can be removed without
encountering any technical difficulties.

To state the results concerning indexing two-dimensional strings, let the dimensions of the text be $H\times W$, with the total size $n=HW$,
and the dimensions of the pattern be $h\times w$, with the total size $m=hw$. Giancarlo~\cite{DBLP:journals/siamcomp/Giancarlo95} introduced the
LSuffix tree of a matrix, based on a linearisation of a two-dimensional string (similar concept has been also used by Amir and
Farach~\cite{DBLP:journals/ipl/AmirF92}). This allowed him to design an index of (asymptotically optimal) size $\Oh(n)$ that
supports queries in $\Oh(m)$ time (for simplicity, we restate the bounds for constant alphabets), but only if the pattern is a square
matrix. This assumption was in fact inherent to his approach, as otherwise it is not clear how to linearise the strings.
For the general case, he designed another structure, called the submatrix tree~\cite{DBLP:conf/wads/Giancarlo93}, of size
$\Oh(\min(H,W)n)$ that supports queries in $\Oh(m)$ time. In the same paper, he defined an abstract notion of an index for
a two-dimensional text, and proved that the size of any such index must be $\Omega(\min(H,W)n)$, making his
construction essentially optimal (perhaps up to a logarithmic factor). Consequently, subsequent work focused on the case of
square patterns~\cite{DBLP:journals/algorithmica/KimNSP11,DBLP:journals/ipl/ChoiL97,DBLP:journals/algorithmica/NaGP07,DBLP:journals/jc/GiancarloG99a,DBLP:journals/jc/GiancarloG99,DBLP:conf/spaa/GiancarloG93,DBLP:conf/icalp/GiancarloG95,DBLP:conf/ifipTCS/FredrikssonNU00}.

\paragraph{Our contribution.} We revisit the complexity of indexing a two-dimensional text for the general case of rectangular
patterns. We observe that, in fact, the abstract notion of an index, as defined by Giancarlo~\cite{DBLP:conf/wads/Giancarlo93},
is somewhat restricted: he defines what it means for one pattern to be a prefix of another, and then requires that the index
has a form of a compacted tree, with every node corresponding to some pattern occurring in the text, and the parent of each node
corresponding to its prefix. This is consistent with the notion of one-dimensional suffix trees, but suffix trees are not the only known
indexing structures. For example, suffix arrays~\cite{DBLP:journals/siamcomp/ManberM93}
use $\mathcal{O}(n)$ space and allow retrieving the occurrences in $\Oh(m+\log n+k)$ time, even though they are not based
on storing a compacted trie (although they are of course related to suffix arrays).
We show that a very simple idea allows us to obtain the following result.

\begin{restatable}{theorem}{restateMainTheorem}\label{main theorem}
For a two-dimensional text of size~$n$, there is an $\Oh(n \log n)$-space data structure that, given a two-dimensional pattern of size~$m$, reports all~$k$ occurrences of the pattern in the text in time~$\Oh(m + k \log^\varepsilon n)$, for any constant~$\varepsilon > 0$.
\end{restatable}

\noindent We also show that the structure from the above theorem can be constructed in $\tilde{\Oh}(n)$ time.

\paragraph{Computational model.}
In the above theorem we assume the standard word RAM model with words of logarithmic (in the size of the input) length. Basic arithmetic
operations on such words (and indirect addressing) are assumed to take constant time. Each character is assumed to fit in a single machine word.
We measure the size of our data structures in the number of words.

\section{Preliminaries}

\paragraph{One-dimensional strings.}
A (one-dimensional) string is a sequence of characters. We index positions in a string starting from~1.
For a string~$S$ of length~$n$, we write $S[i]$ to denote its $i$-th character and $S[i \dd j]$ to denote the substring spanning positions~$i$ through~$j$, inclusive.
If only one endpoint is provided, we interpret the interval as a prefix or suffix: $S[\dd i] := S[1 \dd i]$ is the prefix of length~$i$; $S[i \dd] := S[i \dd n]$ is the suffix starting at position~$i$.

\paragraph{Two-dimensional strings.}
We define two-dimensional strings as rectangular arrays of characters.
We refer to the total number of characters in a two-dimensional string as its \emph{size}.
Rows and columns are indexed starting from~1.
For a two-dimensional string $S$ we write $S[i]$ to denote its $i$-th row, interpreted as a one-dimensional string.
Consequently, $S[i][j]$ denotes the character in the $i$-th row and $j$-th column.

\paragraph{Meta-characters.}
In our algorithm, we reduce two-dimensional indexing to a collection of problems concerning one-dimensional strings.
This is done by treating fixed-length fragments of rows of the two-dimensional text as atomic symbols, which
we refer to as \emph{meta-characters}. Formally, for a two-dimensional string~$S$ and a fixed width~$w$,
a meta-character is a substring of the form $S[i][j \dd j + w - 1]$.
To effectively operate on such meta-characters, we will represent each of them by an integer from $[\poly(n)]$ that fits
inside a single machine word\footnote{It is in fact enough that it fits in a constant number of machine words.}, called
the identifier. The identifiers of two meta-characters will be different if and only if the meta-characters themselves are different.

\paragraph{Compacted trie.}
We use the standard compacted trie, also known as a compressed trie or Patricia trie, for storing a set
of strings. The strings consist of either characters or meta-characters. In either case, we assume that the symbols
fit in a single machine word. A compacted trie built for a set of $k$ strings is of size $\Oh(k)$, as we collapse
maximal paths consisting of nodes with exactly one child into a single edge, thus the number of inner nodes
is strictly smaller than the number of leaves. The remaining nodes are called explicit, while the dissolved nodes are called implicit.

\paragraph{Tools.}
We will need a few data structures. The first is called a deterministic dictionary.
\begin{theorem}[\cite{DBLP:journals/jacm/FredmanKS84}]\label{static dictionary}
Given a set $S$ of $n$ keys, we can build a dictionary structure of size $\Oh(n)$ that allows constant-time access to any element of $S$
(and its associated information).
\end{theorem}

\noindent The second is a range reporting structure (we note that other trade-offs between the space and reporting time are possible, but we state
only one of them to avoid clutter).
\begin{theorem}[{\cite[Theorem 2.1]{Chan2011}}]\label{fact:chan}
Given a collection of~$n$ points in~$\{1, \dots, n\}^2$, we can build a structure of size $\Oh(n)$ that, given an axis-aligned rectangle,
reports all~$k$ points inside it in time~$\Oh((1 + k) \log^\varepsilon n)$, for any constant~$\varepsilon > 0$.
\end{theorem}

\noindent Finally, we need to implement \emph{prefix search} on a compacted trie.
That is, given a query string, we traverse the trie to determine whether it occurs as a prefix of any stored string.
If so, we return the corresponding node (which may be implicit); otherwise, we report that no such prefix match exists.

\begin{lemma}\label{fact:compacted-trie}
A compacted trie storing $k$ strings uses $\Oh(k)$ additional space (on top of the strings themselves).
Assuming constant-time read-only random access to any character, prefix search for a query string of length $m$ takes $\Oh(m)$ time.
\end{lemma}
\begin{proof}
To implement prefix search, we start at the root and descend following the appropriate edge. This requires storing
all the edges outgoing from an explicit node in a dictionary structure, which we implement with \Cref{static dictionary}.
Since the total degree over all explicit nodes is $\Oh(k)$, the combined space used by all dictionaries is $\Oh(k)$.
For implicit nodes, we only verify whether the subsequent characters of the query string are the same as on the edge,
using constant-time read-only random access to one of the stored strings. Then, the time per each character of the
query string is constant, so $\Oh(m)$ overall.
\end{proof}

\section{Two-Dimensional Index Construction}

In this section, we present our data structure for indexing a two-dimensional text and prove \Cref{main theorem}.
The idea behind our approach is to reduce the two-dimensional pattern matching problem to a collection of one-dimensional problems.
Each of these problems can then be solved efficiently by plugging in range reporting structures, as usual in the one-dimensional setting.
Throughout this section, we write~$T$ to denote the $H \times W$ input text and~$P$ to denote the $h \times w$ query pattern,
with~$n = HW$ and~$m = hw$ denoting their respective sizes.

The data structure consists of two symmetric components, designed to handle the query depending on whether the pattern is taller or wider.
Given an $h \times w$ pattern, we use one component if $h \ge w$ and the symmetric one (with rows and columns interchanged) otherwise.
We describe how to handle tall patterns, i.e., those with $h \ge w$, as the other case is fully symmetric.
To handle patterns of different widths, we store a separate index for each possible pattern width.
More precisely, for each width $w \in \{1, \dots, W\}$, we construct a dedicated data structure that supports searching for
patterns of fixed width $w$ and any height $h \ge w$.

For a fixed width $w$, we reduce the two-dimensional matching problem to a one-dimensional problem.
Specifically, we interpret the pattern as a one-dimensional string of length~$h$, where the $i$-th symbol is a \emph{meta-character} corresponding to the $i$-th row $P[i]$ of the pattern, that is, a contiguous block of~$w$ characters from that row.
As explained earlier, we will assign identifiers to these meta-characters, so that we can treat them as integers from $[\poly(n)]$.
We first explain how to ensure that we can access the identifier of any meta-character in the pattern and the text in constant time,
after $\Oh(n\log n)$ space preprocessing of the text, and then define the one-dimensional problem.

\paragraph{Encoding the text.}
The construction follows the classical Karp-Miller-Rosenberg (KMR) approach~\cite{KMR}, and is based on assigning integer identifiers to all fragments of the form $T[i][j \dd j + 2^k - 1]$, for all $k$ such that $2^k \le W$.
For each such fragment, we define its identifier as the lexicographic rank of its substring among all distinct substrings of the same length, and store this value. Then, the identifier of $T[i][j \dd j+w-1]$ consists of the identifiers of two
overlapping fragments of length~$2^{\lfloor \log w \rfloor}$ that together cover the whole $T[i][j \dd j+w-1]$.
Since there are $\Oh(n)$ fragments per length and $\Oh(\log W)$ relevant lengths, the total space required is~$\Oh(n \log n)$.

\paragraph{Encoding the pattern.}
To encode a pattern row~$P[i]$ of width~$w$, we represent the corresponding meta-character by the pair of identifiers of its prefix and suffix of length~$2^k$, where $k = \lfloor \log w \rfloor$.
To support this, during text preprocessing, we collect all distinct substrings $T[i][j \dd j + 2^k - 1]$ and store them in a compacted trie, where each leaf is labeled with the lexicographic rank among all substrings of $T$ of length~$2^k$.
We build a separate compacted trie for each relevant value of~$k$; since each trie requires~$\Oh(n)$ space and there are~$\Oh(\log W)$ values of~$k$, the total space used by all the compacted tries is~$\Oh(n \log n)$.
At query time, we extract the $2^k$-length prefix and suffix of~$P[i]$ and search for them in the corresponding compacted trie to obtain their identifiers.
If either substring is not found, we report that the pattern does not occur in the text. Otherwise, in $\Oh(m)$ time we obtain the sought identifier
for each row of the pattern.

\paragraph{One-dimensional problems.}
To complete the reduction, we must identify all occurrences of the one-dimensional meta-character pattern within the two-dimensional text.
To that end, we consider all windows of~$w$ consecutive columns.
Each such $w$-column strip defines a one-dimensional text of length~$H$, where the $j$-th character is a meta-character corresponding to the block of $w$ consecutive characters in the $j$-th row of the strip.
Formally, to define the one-dimensional texts corresponding to each $w$-column strip, we fix a horizontal offset~$i \in \{1, \dots, W - w + 1\}$ and define a one-dimensional text of length $H$ consisting of the meta-characters $T[j][i \dd i + w - 1]$ for all positions~$j\in\set{1, \dots, H}$.
This gives a collection of $W - w + 1$ one-dimensional texts over the same meta-character alphabet.
Each occurrence of the meta-character pattern in one of these one-dimensional texts corresponds one-to-one to a two-dimensional occurrence of the original pattern in the text.

In each of these one-dimensional instances, the derived meta-character pattern has length~$h \ge w$.
We exploit this assumption when designing a (simple) indexing structure in \Cref{main lemma proof}.

\begin{restatable}{lemma}{restateMainLemma}\label{main lemma}
For a parameter $w$ and a collection of one-dimensional texts of total length $n$, there is a data structure requiring $\Oh(n / w)$ additional space (on top of the strings) that, given a one-dimensional pattern of length $h \ge w$, reports all $k$ occurrences of the pattern in the texts in time $\Oh(hw + (w + k) \log^\varepsilon n)$, for any constant $\varepsilon > 0$.
\end{restatable}

We could apply \Cref{main lemma} for each width $w \in \{1, \dots, W\}$ to handle patterns of width $w$ and height $h\ge w$. The total additional
space is bounded by $\sum_{w = 1}^W \Oh(n / w) = \Oh(n \log n)$, matching the requirements of \Cref{main theorem}.
However, when $h$ is smaller than $\log^\varepsilon n$, it holds that $w \log^\varepsilon n$ in the query time is larger than $m$.
We thus handle the case of small $w$ separately as follows.

\begin{lemma}\label{lemma:suffix-tree}
For a collection of one-dimensional texts of total length~$n$, there is an $\Oh(n)$-space data structure that, given a one-dimensional pattern of any length~$h$, reports all~$k$ occurrences of the pattern in the texts in time~$\Oh(h + k)$.
\end{lemma}
\begin{proof}
We store all suffixes of all input texts, each terminated with a distinct delimiter, in a compacted trie implemented with \Cref{fact:compacted-trie}.
The total number of suffixes is $\Oh(n)$, so the total required space is $\Oh(n)$.
Querying for a pattern of length $h$ is implemented by performing a prefix search for the pattern in the trie and reporting all leaves in
the corresponding subtree.
\end{proof}

We apply \Cref{lemma:suffix-tree} for every width $w\leq \log n$, and otherwise use \Cref{main lemma}.
The total space is still $\Oh(n\log n)$.
To bound the query time, we observe that whenever we use \Cref{lemma:suffix-tree} the query time is $\Oh(h+k)=\Oh(m+k)$,
and whenever we use \Cref{main lemma} we can assume $w>\log n$ so the query time is $\Oh(hw+(w+k)\log^{\epsilon}n)=
\Oh(m+k\log^{\epsilon}n)$.

\section{One-Dimensional Index for Long Patterns}
\label{main lemma proof}

In this section, we prove \Cref{main lemma} by designing an indexing structure for the one-dimensional setting that takes advantage of the
assumption that the pattern is sufficiently long. We briefly comment that such a setup can be seen as related to locally consistent anchors,
and such an assumption has been used in the literature~\cite{DBLP:journals/pvldb/AyadLP23}. In our case, a simple and self-contained
approach based on defining some compacted tries and storing a range reporting structure is enough, though.
The idea of using range searching to solve indexing is, of course, quite common in the literature~\cite{DBLP:conf/birthday/Lewenstein13}.

\paragraph{Setup.}
We assume that the input texts allow constant-time read-only random access to the individual characters, and each of them
fits in a single machine word. Throughout this section, we fix a parameter~$w$, corresponding to the lower bound on the length of the pattern.
Texts shorter than $w$ can be discarded, as they cannot contain an occurrence of a pattern of length~$h \ge w$.

\paragraph{Preprocessing.}
For each text~$T$, we define a collection of \emph{cuts}: positions between characters spaced at regular intervals of length~$w$.
Formally, we insert a cut at every position~$i$ such that~$i \equiv 0 \pmod w$, including~$i = 0$.
Each cut partitions~$T$ into a prefix~$T_1 = T[\dd i]$ and a suffix~$T_2 = T[i + 1 \dd]$, where either may be empty.
We associate the cut with the pair~$(T_1, T_2)$, which serves as its representation.
We organize the collection of cuts using a two-dimensional range reporting structure. Let~$\mathcal{T}_1$ denote the set of all prefixes~$T_1$ reversed, and let~$\mathcal{T}_2$ denote the set of all suffixes~$T_2$ extracted from the cuts. We build two compacted tries:
\begin{itemize}
    \item $\mathcal{S}_1$, storing the strings in~$\mathcal{T}_1$,
    \item $\mathcal{S}_2$, storing the strings in~$\mathcal{T}_2$.
\end{itemize}
To ensure that each string from $\mathcal{T}_1$ and from $\mathcal{T}_2$ corresponds to a unique leaf in the respective trie, we
conceptually prepend and append a distinct delimiter character to each prefix and suffix, respectively.

Next, we assign a pre-order number to each leaf of~$\mathcal{S}_1$ and~$\mathcal{S}_2$ via a depth-first traversal.
Each cut then defines a point~$(x, y) \in \mathbb{Z}^2$, where:
\begin{itemize}
    \item $x$ is the pre-order number of the leaf in~$\mathcal{S}_1$ corresponding to the reversal of~$T_1$,
    \item $y$ is the pre-order number of the leaf in~$\mathcal{S}_2$ corresponding to~$T_2$.
\end{itemize}
This yields a collection of~$\Oh(n / w)$ such points, one per cut. We store these points in an instance of \Cref{fact:chan}.

\paragraph{Answering a query.}
Given a pattern~$P$ of length~$h \ge w$, our goal is to report all of its occurrences in the input texts.
Since the pattern has length at least~$w$, any occurrence must span at least one cut position from the collection defined for the texts.
We iterate over all positions at which such a cut could intersect an occurrence of the pattern.
We refer to each such position within the pattern as an \emph{anchor}, and consider
only the first $w$ possible anchors, to ensure that every occurrence of~$P$ is reported exactly once.

Specifically, we consider all positions~$j \in \{0, 1, \dots, w - 1\}$ where the pattern may be anchored.
Each such position induces a partition of~$P$ into a prefix~$P_1 = P[\dd j]$ and a suffix~$P_2 = P[j + 1 \dd]$.
The task is to find and report all cuts in the texts that partition some text into~$(T_1, T_2)$ such that~$P_1$ is a suffix of~$T_1$ and~$P_2$ is a prefix of~$T_2$.
This is implemented as follows.
First, we extract the corresponding prefix-suffix pair~$(P_1, P_2)$ from the pattern.
We then locate the (possibly implicit) node~$v_1$ in~$\mathcal{S}_1$ corresponding to the reversal of~$P_1$, and the node~$v_2$ in~$\mathcal{S}_2$ corresponding to~$P_2$.
If either node does not exist, the current position cannot yield any occurrences and is skipped.
Otherwise, let~$[\ell_1, r_1]$ and~$[\ell_2, r_2]$ be the ranges of pre-order numbers of leaves in the subtrees of~$v_1$ and~$v_2$, respectively.

The problem now reduces to a two-dimensional orthogonal range reporting query: report all stored points~$(x, y)$ that lie within the rectangle~$[\ell_1, r_1] \times [\ell_2, r_2]$.
Each reported point corresponds to a valid cut that certifies an occurrence of the pattern.
Querying the instance of \Cref{fact:chan} yields the query bound claimed in \Cref{main lemma}.

\section{Construction Time}\label{sec:construction-time}

We briefly discuss the preprocessing time needed to build the index from \Cref{main theorem}.
Recall that we use two external components as black boxes:
\begin{itemize}
	\item The deterministic dictionaries of \Cref{static dictionary}, which for a set of $n$ keys can be constructed in $\Oh(n(\log\log n)^2)$ time \cite{ruzic2008dictionaries}.
	\item The orthogonal range reporting structure of \Cref{fact:chan}, the construction time of which is not stated in the original description of \cite{Chan2011}.
		We claim that given $n$ points, it can be built in $\tilde{\Oh}(n)$ time, which suffices for our purposes.
\end{itemize}

Because the exact logarithmic factor depends on the two black-box components mentioned above, and to keep the description less technical, we state the bound in the $\tilde{\Oh}(\cdot)$ notation.

\begin{theorem}\label{thm:construction-time}
For a two-dimensional text of size $n$, the data structure of \Cref{main theorem} can be constructed in $\tilde{\Oh}(n)$ time.
\end{theorem}

In what follows we revisit the components of the two-dimensional index and describe their efficient construction, proving \Cref{thm:construction-time}.

As described previously, for all $k$ with $2^k \le W$ we assign identifiers to all substrings $T[i][j \dd j+2^k-1]$ using the standard doubling (Karp–Miller–Rosenberg) scheme with radix sorting \cite{KMR}.
This takes $\Oh(n\log n)$ time in total.
We then store all substrings in a compacted trie, one trie per $k$.
To construct these tries, as well as the tries $\mathcal{S}_1$ and $\mathcal{S}_2$ used for 1D indexing, we need two lemmas.
The first lemma is not new (it is essentially the same approach as the one used to construct a Cartesian tree in linear time),
but we present the proof for completeness.

\begin{lemma}\label{lemma:trie-construction}
	A compacted trie storing $k$ strings can be built in $\tilde{\Oh}(k)$ time, assuming that the longest common extension (LCE) of any two positions in the strings can be answered in $\tilde{\Oh}(1)$ time.
\end{lemma}
\begin{proof}
We start with lexicographically sorting the strings. This requires $\Oh(k\log k)$ comparisons, and each comparison can be implemented with
an LCE query and then retrieving the characters at the subsequent position in both strings in $\tilde{\Oh}(1)$ time. Thus, after $\tilde{\Oh}(k)$
time we have the strings in the lexicographically sorted order $s_{1},s_{2},\ldots,s_{k}$. Next, we iterate over $i=1,2,\ldots,k$ while maintaining
the compacted trie $T_{i}$ storing $s_{1},s_{2},\ldots,s_{i}$. The trie is initially empty. We observe that the trie $T_{i+1}$ can be obtained from
$T_{i}$ by inserting $s_{i+1}$ as the rightmost leaf (we assume that each string $s_{i}$ ends with a unique special character). This possibly
requires splitting an edge by creating a new explicit node, and then creating a new edge leading to the new leaf. By calculating the LCE of
$s_{i}$ and $s_{i+1}$ we know the depth at which we should attach the new edge. We traverse the path from the leaf corresponding to $s_{i}$
up to find either the explicit node to which we should attach the new edge or an edge that should be split by creating a new explicit node to
which we should attach the new edge. The traversal takes constant time per considered edge, as we can maintain the current depth. The crucial
observation is that the total time of all the traversals is only $\Oh(k)$, as each traversed edge corresponds to a unique edge of the resulting
compact trie $T_{k}$. Thus, the total time is as claimed.
\end{proof}

The second lemma is standard, e.g. by constructing the suffix array~\cite{DBLP:journals/jacm/KarkkainenSB06} and augmenting it with an RMQ data structure~\cite{DBLP:conf/cpm/FischerH06}.

\begin{lemma}\label{lemma:lce-construction}
	A string of length $n$ can be preprocessed in $\tilde{\Oh}(n)$ time to support LCE queries in $\tilde{\Oh}(1)$ time.
\end{lemma}

By \Cref{lemma:trie-construction}, to efficiently build a compacted trie storing all distinct substrings $T[i][j \dd j+2^k-1]$, it suffices to augment the row fragments with LCE support, which we do by concatenating all rows of $T$ and using \Cref{lemma:lce-construction}.

To efficiently build the compacted tries used in the 1D indexes of \Cref{main lemma,lemma:suffix-tree}, we also build an LCE oracle for the $w$-column strip texts (and their reversals).
We cannot apply \Cref{lemma:lce-construction} separately for every $w$, and instead only apply it on strips of width $2^{k}$ as follows.
For each $k$ and each starting column $i \in \{1,\dots,W-2^k+1\}$, the $2^k$-column strip $T[\cdot][i \dd i+2^k-1]$ induces a one-dimensional string of length $H$ over meta-characters.
We concatenate all such strings for a fixed $k$ and use \Cref{lemma:lce-construction}, enabling LCE queries between any two positions in any strings induced by $2^k$-column strips.
The total length per $k$ is $\Oh(n)$, so the total preprocessing time is $\tilde{\Oh}(n)$.
Next, to handle an LCE query between some positions from strings $S_{i}$ and $S_{i'}$ induced by $w$-column strips $T[\cdot][i \dd i + w - 1]$ and $T[\cdot][i' \dd i' + w - 1]$, for an arbitrary width $w$, we reduce it to two LCE queries on $2^{k}$-column strips as follows.
Let $k:=\lfloor \log w \rfloor$ and let $L_i, R_i$ denote the strings induced by $2^k$-column strips $T[\cdot][i \dd i + 2^k - 1]$ and $T[\cdot][i+w-2^k \dd i+w-1]$, covering $T[\cdot][i \dd i + w - 1]$. 
We define $L_{i'}$ and $R_{i'}$ analogously to cover $T[\cdot][i' \dd i' + w - 1]$.
Then, the length of the longest common prefix of $S_{i}[j\dd]$ and $S_{i'}[j'\dd]$ is the minimum of the length of the longest common prefix of $L_{i}[j\dd]$ and $L_{i'}[j'\dd]$ and the length of the longest common prefix of $R_{i}[j\dd]$ and $R_{i'}[j'\dd]$. Both lengths can be found with the already constructed oracle.

We apply the same preprocessing on the vertically reversed text to support the reversed-LCE queries used to construct tries storing reversed prefixes.
Adding the near-linear constructions of the dictionaries and the range-reporting structures yields the overall $\tilde{\Oh}(n)$ preprocessing time claimed in \Cref{thm:construction-time}.

\section{Acknowledgements}

We would like to thank Itai Boneh and Panagiotis Charalampopoulos for participating in the initial discussion.

\bibliographystyle{plainurl}
\bibliography{biblio}

\end{document}